\pdfoutput=1
\documentclass[]{article}  
\usepackage{url,float}
\usepackage{graphicx}
\usepackage{amsmath}
\usepackage{amsfonts}
\usepackage{amssymb}
\usepackage{latexsym}
\usepackage{mathtools}
\usepackage{color}
\usepackage[boxruled]{algorithm2e}

\newcommand{\hide}[1]{}

\newcommand{\ABox}{
\raisebox{3pt}{\framebox[6pt]{\rule{6pt}{0pt}}}
}
\newenvironment{proof}{{\bf Proof:}}{\hfill\ABox}

\newtheorem{theorem}{{\bf Theorem}}
\newtheorem{cor}{Corollary}
\newtheorem{lemma}{Lemma}
\newtheorem{prop}{Proposition}

\newcommand{\lemlab}[1]{\label{lemma:#1}}
\newcommand{\thmlab}[1]{\label{thm:#1}}
\newcommand{\proplab}[1]{\label{prop:#1}}
\newcommand{\corlab}[1]{\label{cor:#1}}

\newcommand{\figlab}[1]{\label{fig:#1}}
\newcommand{\seclab}[1]{\label{sec:#1}}

\newcommand{\lemref}[1]{\ref{lemma:#1}}
\newcommand{\thmref}[1]{\ref{thm:#1}}
\newcommand{\corref}[1]{\ref{cor:#1}}

\newcommand{\secref}[1]{\ref{sec:#1}}
\newcommand{\figref}[1]{\ref{fig:#1}}

\def\a{{\alpha}}
\def\b{{\beta}}
\def\q{{\theta}}

\def\g{{\gamma}}

\newcommand{\squeezelist}{\setlength{\itemsep}{0pt}}

\usepackage{enumitem}

\title{%
Polar Zonohedra Edge-Unfold to Nets
} 

\author{%
Joseph O'Rourke%
  \thanks{Departments of Computer Science and of Mathematics, Smith College, Northampton, MA
      01063, USA.
      \protect\url{jorourke@smith.edu}.}
}

\date{\today}

\begin{document}
\maketitle

\begin{abstract}
This note proves that
every polar zonohedron has an edge-unfolding to a non-overlapping net.
\end{abstract}

\section{Introduction}
\seclab{Introduction}
What has become known as D{\"u}rer's Problem~\cite{o-dp-13}
asks whether or not every convex polyhedron $P$ has an edge-unfolding---a spanning tree
of the $1$-skeleton which when cut, unfolds the surface of $P$ to a planar
simple polygon known as a \emph{net}.
A key requirement is that a net does not self-overlap.\footnote{
Even though redundant, to block misinterpretation we sometimes call it
a \emph{non-overlapping net}.}
Despite considerable effort since Shephard first posed the question formally~\cite{s-cpcn-75},
there are only finite classes of polyhedra  
($5$ Platonic solids, $13$ Archimedean solids, $92$ Johnson solids),
and a few infinite classes of polyhedra known to edge-unfold to a net:
prisms, domes~\cite[p.~323]{do-gfalop-07}, 
prismoids~\cite{o-upwo-01}~\cite[p.~325]{do-gfalop-07}, 
nested prismatoids~\cite{radons2021edge}.

A \emph{zonohedron} is a centrally symmetric convex polyhedron that may be defined
as the Minkowski sum of a set of line segments.
A \emph{polar zonohedron} is formed when the generating segments constitute an
equally spaced and equal (unit) length star (or ``umbrella") at one of its two poles.
It is symmetric about the line through the two poles;
we'll assume this line is vertical.
It is also symmetric about a plane orthogonal to and through the midpoint of the poles line.
All of its edges have the same length;
all of its faces are rhombs.
A polar zonohedron is determined by two parameters:
$n$, the number of rhombs incident to each pole,
and $\q \in (0,\pi/2)$, the angle that measures the degree of closure at the poles.
Thus there is a continuum of polar zonohedra.
Small $\q$ produce pancake-like polyhedra, and $\q$ near $\pi/2$ produce thin cigar-shaped polyhedra.
See Fig.~\figref{PZ_Unf_n8_50deg}.
\begin{figure}[htbp]
\centering
\includegraphics[width=1.1\linewidth]{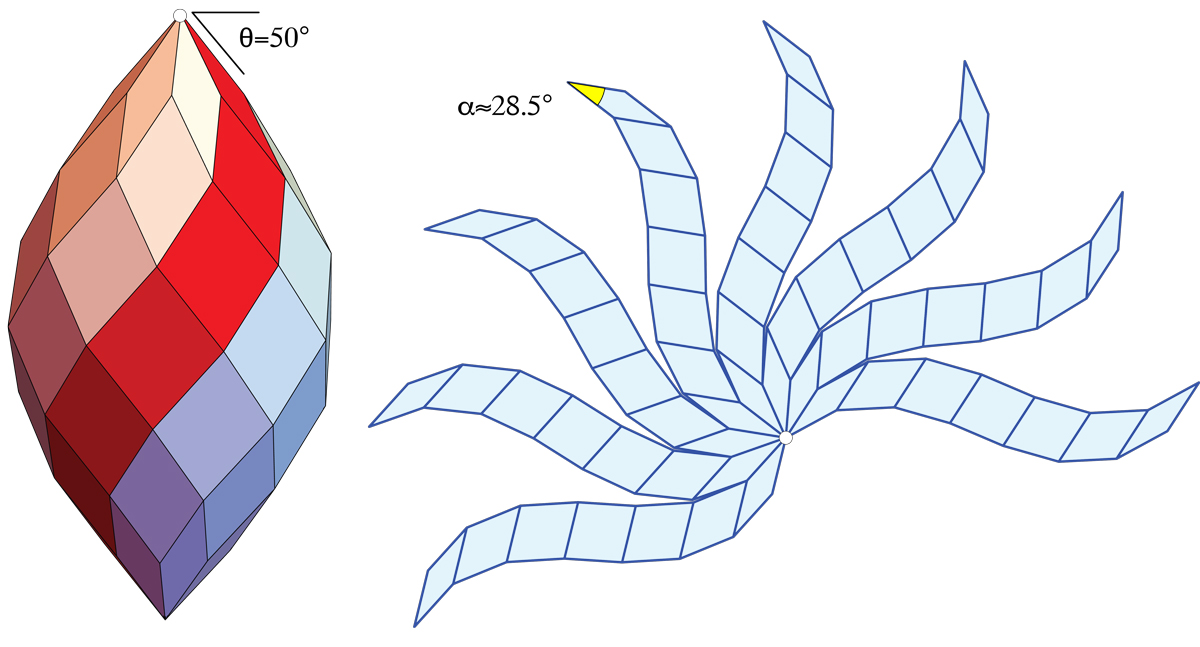}
\caption{$n=8, \q=50^\circ$. 
One zone $\mathcal{Z}$ of $(n-1)=7$ rhombs is marked in red.}
\figlab{PZ_Unf_n8_50deg}
\end{figure}

The aim of this note is to prove this theorem:
\begin{theorem}
\thmlab{net}
Every polar zonohedron has an edge-unfolding to a (non-overlapping) net.
\end{theorem}

The unfolding is illustrated in Fig.~\figref{PZ_Unf_n8_50deg}.
It unfolds each of the $n$ zones, each joined at the image $o$ of the north pole.
A \emph{zone} of a polar zonohedron is a collection of rhombs all sharing
a set of parallel edges. There are $n$ such congruent zones comprising the surface,
each containing $(n-1)$ rhombs. 
The edges cut between each pair of adjacent zones form
what George Hart calls a ``surface helix."
Two further examples are shown in Figs.~\figref{PZ_Helix_n20_theta_05}
and~\figref{Octopus_n12_q05}

I rely on Hart's thorough exposition in~\cite{HartPZono},
as well as employing his Mathematica software.
See also~\cite{RussellTowle}.

\begin{figure}[htbp]
\centering
\includegraphics[width=1.1\linewidth]{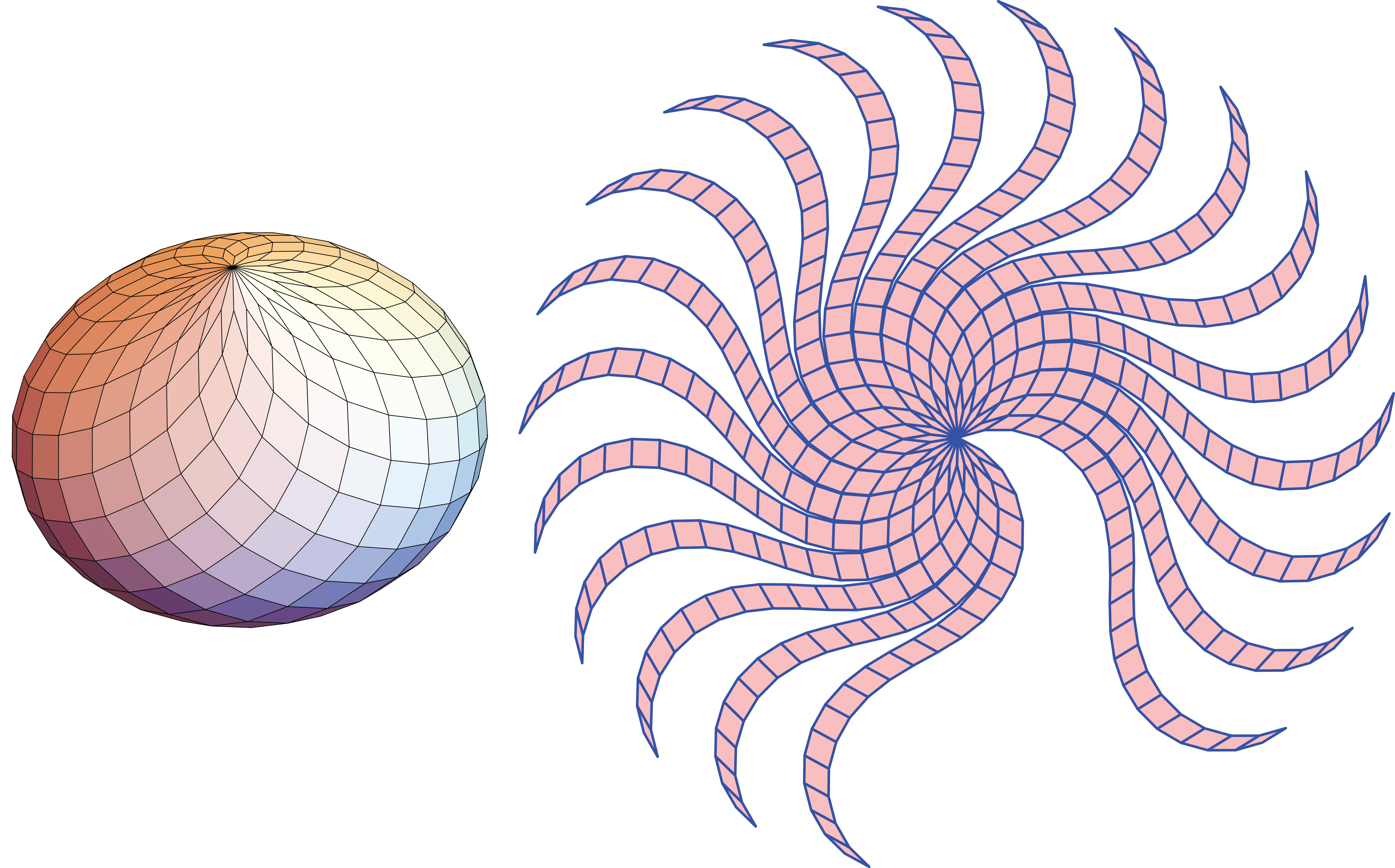}
\caption{$n=20$, $\q = 0.5 \approx 28.6^\circ$, $\a \approx 15.8^\circ$.}
\figlab{PZ_Helix_n20_theta_05}
\end{figure}
\begin{figure}[htbp]
\centering
\includegraphics[width=0.75\linewidth]{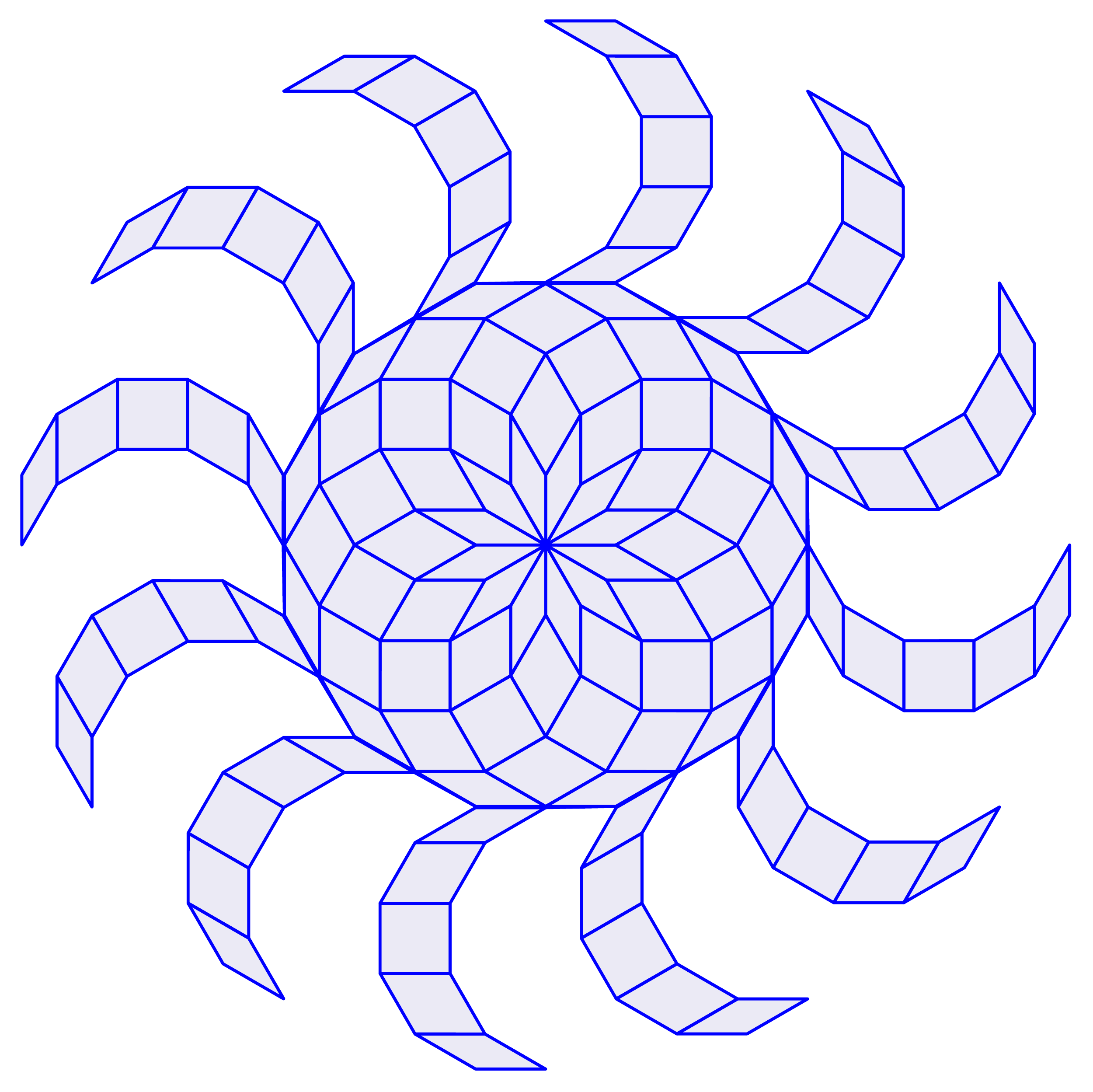}
\caption{$n=12$, $\q = \frac{1}{2}^\circ$.}
\figlab{Octopus_n12_q05}
\end{figure}

\section{Notation}
\seclab{Notation}
Distinguish between a zone $\mathcal{Z}$ in $\mathbb{R}^3$
and that zone unfolded, or developed in the plane as $Z$.
Each $Z$ is composed of $(n-1)$ rhombs:
$R_1,R_2,\ldots,R_{n-1}$.
Let $\a$ be the angle of $R_1$ at the corner that is incident to the (north) pole point.

Let $o$ be the planar image of the north pole.
We focus on two consecutive unfolded zones, because if adjacent zones do not overlap,
then there is no overlap in the entire unfolding.
All zones are congruent in $\mathbb{R}^3$ and therefore congruent  in their planar unfoldings.
Let $Z$ be one zone, and $Z'$ its counterclockwise (ccw) neighbor sharing an edge incident
to the pole $o$.
See Fig.~\figref{TwoZ_n16_x3}.
\begin{figure}[htbp]
\centering
\includegraphics[width=1.1\linewidth]{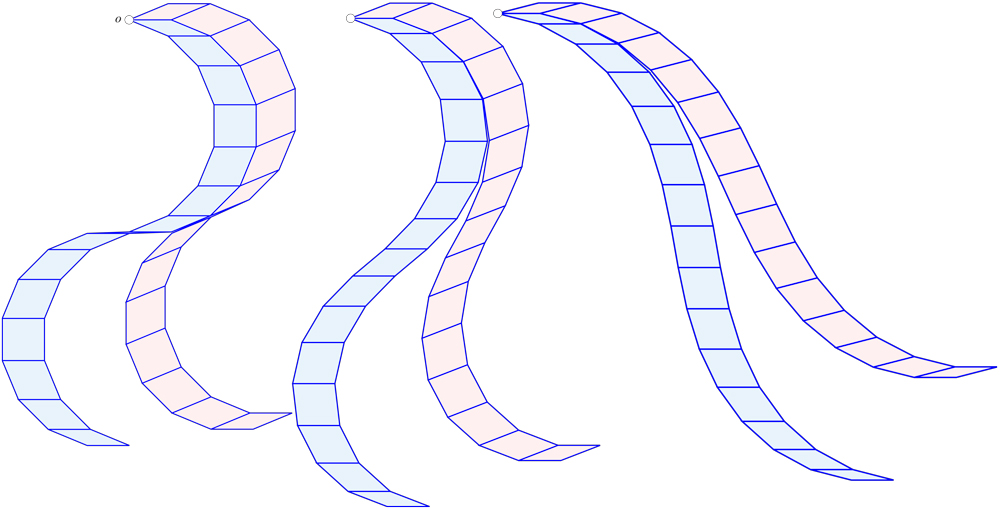}
\caption{$Z$ and $Z'$ for $n=16$. Left to right: $\q=1^\circ,20^\circ,50^\circ$.}
\figlab{TwoZ_n16_x3}
\end{figure}

Arrange $Z$ so that 
the top edge of $R_1$ is horizontal.
Then $Z'$ is $Z$ rigidly rotated about $o$ by $\a$,
because the zones are congruent and angle $\a$ is incident to the pole $o$.
In this orientation all the $n$ unit edges of $Z$ are horizontal,
all the unit edges of $Z'$ at angle $\a$.

\section{Proof Plan}
\seclab{ProofPlan}
Because $Z'$ is just a rigid rotation of $Z$ about $o$, it seems almost obvious that they cannot
overlap.
But I have not found a simple proof.

However, the overall plan of the proof presented here is simple.
First we establish a condition on the shape of $Z$ that guarantees non-overlap.
Let $| \mathcal{C}(r) \cap Z | = \b(r)$ be the angular measure of the arc of the
circle $\mathcal{C}(r)$ centered on $o$ that intersects $Z$.
(This definition is revisited and clarified in the next section.) 
Then the shape condition is that, if $\b(r) \le \a$ for all $r$, then $Z$ and the rotated $Z'$
do not overlap.

The remainder of the proof calculates $\b(r)$ and verifies that it is $\le \a$,
for all $n$ and all $\q$.
It is these calculations that are not straightforward.
The main open problem (Section~\secref{Open}) is to find a proof
that can establish $\b \le \a$ without these calculations.

\section{Overlap Condition}
\seclab{OverlapCondition}
We say that two regions $A$ and $B$ in the plane \emph{overlap} 
if there is a point $p$ strictly interior to both $A$ and $B$.

Let $\mathcal{C}(r)$ be a circle of radius $r$ centered on $o$.
Define $| \mathcal{C}(r) \cap Z |= \b(r)$ as the angle subtended at $o$ by the shortest arc of $\mathcal{C}(r)$
that covers $\mathcal{C}(r) \cap Z$.
(We will drop the $r$ in $\b(r)$ when clear from the context.)
Note here we are allowing for the possibility that $\mathcal{C}(r) \cap Z$ might be
disconnected. And in fact, 
this is necessary, for in extreme situations ($\q \to 0$), this can occur.
See Fig.~\figref{Disconnected_n16_1deg}.
We will revisit disconnected arcs in Section~\secref{FlatRhomb}.
\begin{figure}[htbp]
\centering
\includegraphics[width=1.0\linewidth]{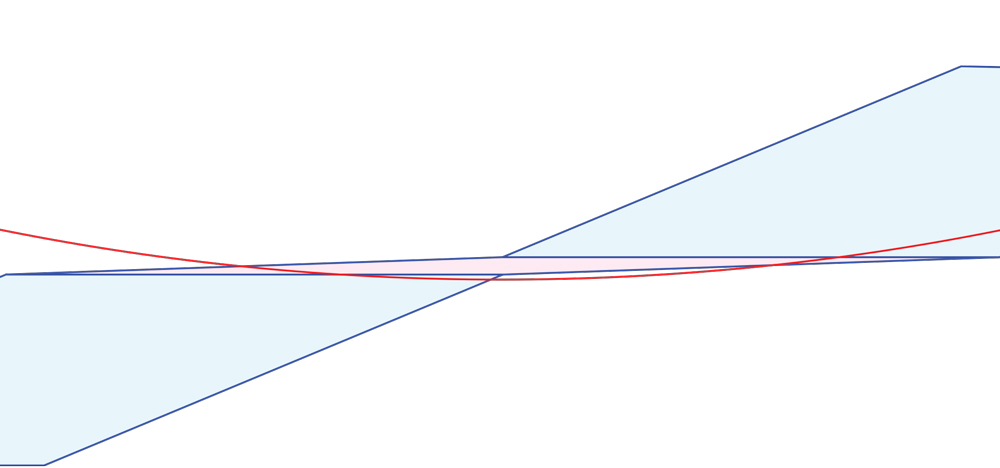}
\caption{ $| \mathcal{C}(r) \cap Z | $ is disconnected.
$n=16$, $\q=1^\circ$. $R_{n/2}=R_8$ is pink.}
\figlab{Disconnected_n16_1deg}
\end{figure}

\begin{lemma}
\lemlab{Overlap}
If $Z$ and $Z'$ overlap at point $p$ on circle $\mathcal{C}(r)$, $r=|op|$,
then $| \mathcal{C}(r) \cap Z | = \b(r) > \a$.
\end{lemma}
\begin{proof}
Suppose for contradiction that $Z$ and $Z'$ overlap at $p$ but 
$| \mathcal{C}(r) \cap Z | = \b \le \a$.
Because $Z'$ is the rotation of $Z$ about $o$ by $\a$,
there is a gap of length $\a-\b \ge 0$ along $\mathcal{C}(r)$,
which separates $\mathcal{C}(r) \cap Z$ to the left (cw) and $\mathcal{C}(r) \cap Z'$ to the right (ccw).
Therefore, $Z$ and $Z'$ are separated or just touch along $\mathcal{C}(r$),
and $p$ cannot exist.
\end{proof}

\medskip
\noindent 
Note that, because $\mathcal{C}(r) \cap Z$ could be disconnected, it is possible
that $\b(r) > \a$ but there is no point $p$ witnessing overlap.
However, if there is a $p$, then it must be that $\b(r) > \a$.

In Fig.~\figref{NoOverlap}, $\mathcal{C}(r) \cap Z$ is the red arc spanning $\b(r) < \a$.
Therefore when this arc is rotated by $\a > \b(r)$ to the blue arc, the red and blue arcs cannot overlap.
\begin{figure}[htbp]
\centering
\includegraphics[width=0.7\linewidth]{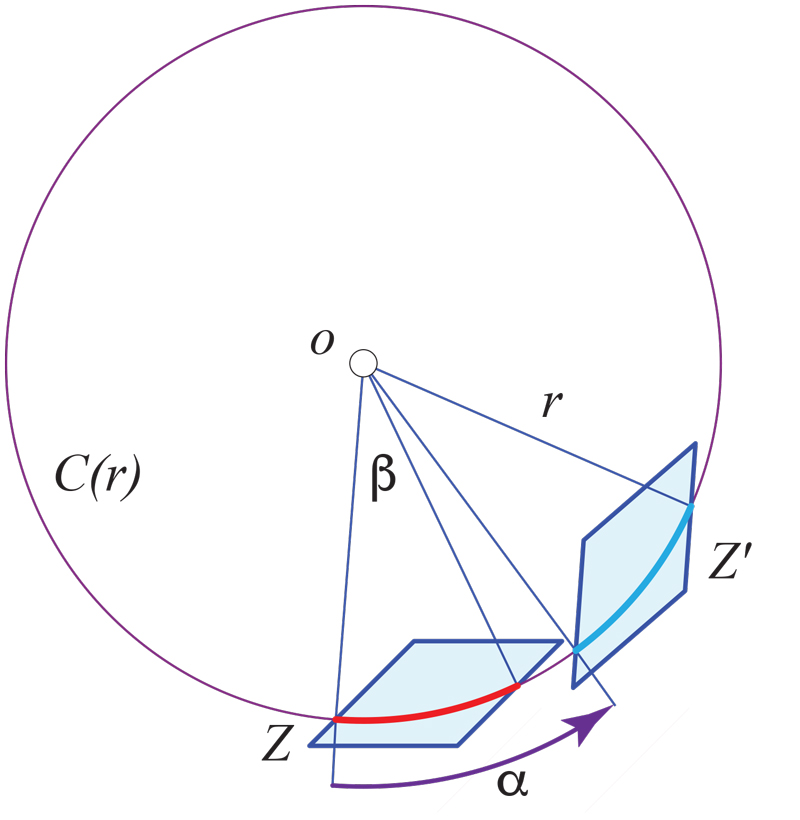}
\caption{$\b \le \a$ implies no overlap.}
\figlab{NoOverlap}
\end{figure}

\begin{cor}
\corlab{NoOverlap}
If $| \mathcal{C}(r) \cap Z | = \b(r) \le \a$ for all $r$, then $Z$ and $Z'$ do not overlap.
\end{cor}
\noindent
The remainder of the proof of Theorem~\thmref{net} concentrates on proving
$\b \le \a$.

\section{Regular Polygon: $\q=0$}
\seclab{S-RegPoly}
There is a sense in which 
the situation when $\q \to 0$ is the most ``dangerous," for then there can
be a flattened rhomb such as that illustrated in Fig.~\figref{Disconnected_n16_1deg}.
When $\q=0$, the zonohedron degenerates to a doubly-covered
regular $n$-gon, with the only curvature at the vertices on the rim.
(See the earlier Fig.~\figref{Octopus_n12_q05}.)
Then the unfolding of a zone $Z$ resembles mirrored $S$-shape composed to two half-regular polygons.
Fig.~\figref{S_RegPoly_n16_1deg}(a) shows an example with $\q=1^\circ$, 
and (b)~the $\q=0$ counterpart.
\begin{figure}[htbp]
\centering
\includegraphics[width=1.1\linewidth]{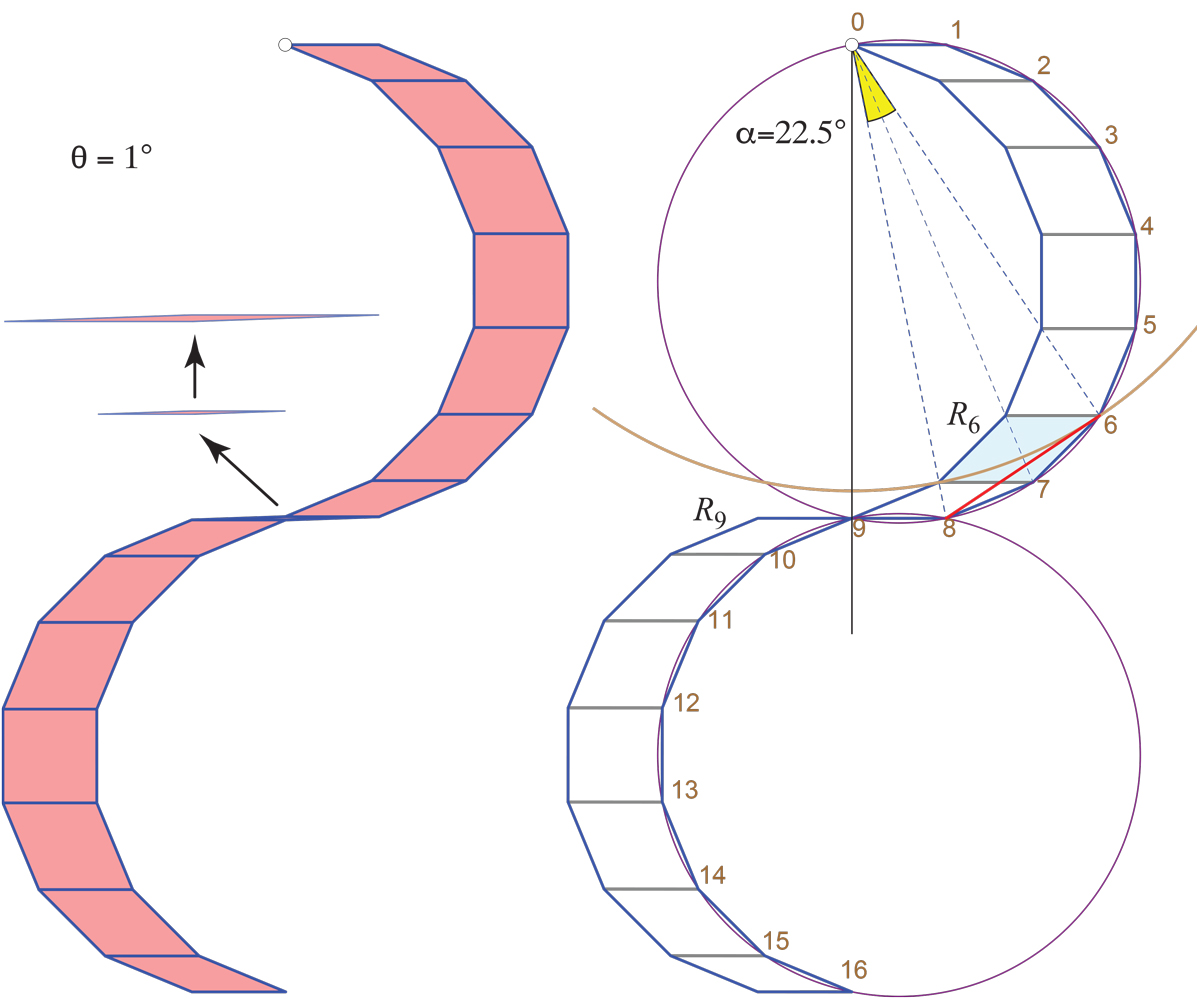}
\caption{$n=16$. (a)~$\q=1^\circ$. Two enlargements of rhomb $R_8$ are shown.
(b)~$\q=0, \a=22.5^\circ$.}
\figlab{S_RegPoly_n16_1deg}
\end{figure}
The $1^\circ$ example~(a) shows that one rhomb $R_8$ is nearly flat;
in~(b) that rhomb collapses to a line segment.

Whether one counts the doubly-covered
regular $n$-gon as a ``polyhedron," it is an interesting case, with
regularities not present when $\q > 0$.

There is a slight difference for $n$ even or odd;
compare Figs.~\figref{S_RegPoly_n16_1deg}
and~Fig.~\figref{S_RegPoly_n17_1deg}.
We will concentrate on $n$ even, leaving $n$ odd to remarks.
\begin{figure}[htbp]
\centering
\includegraphics[width=1.1\linewidth]{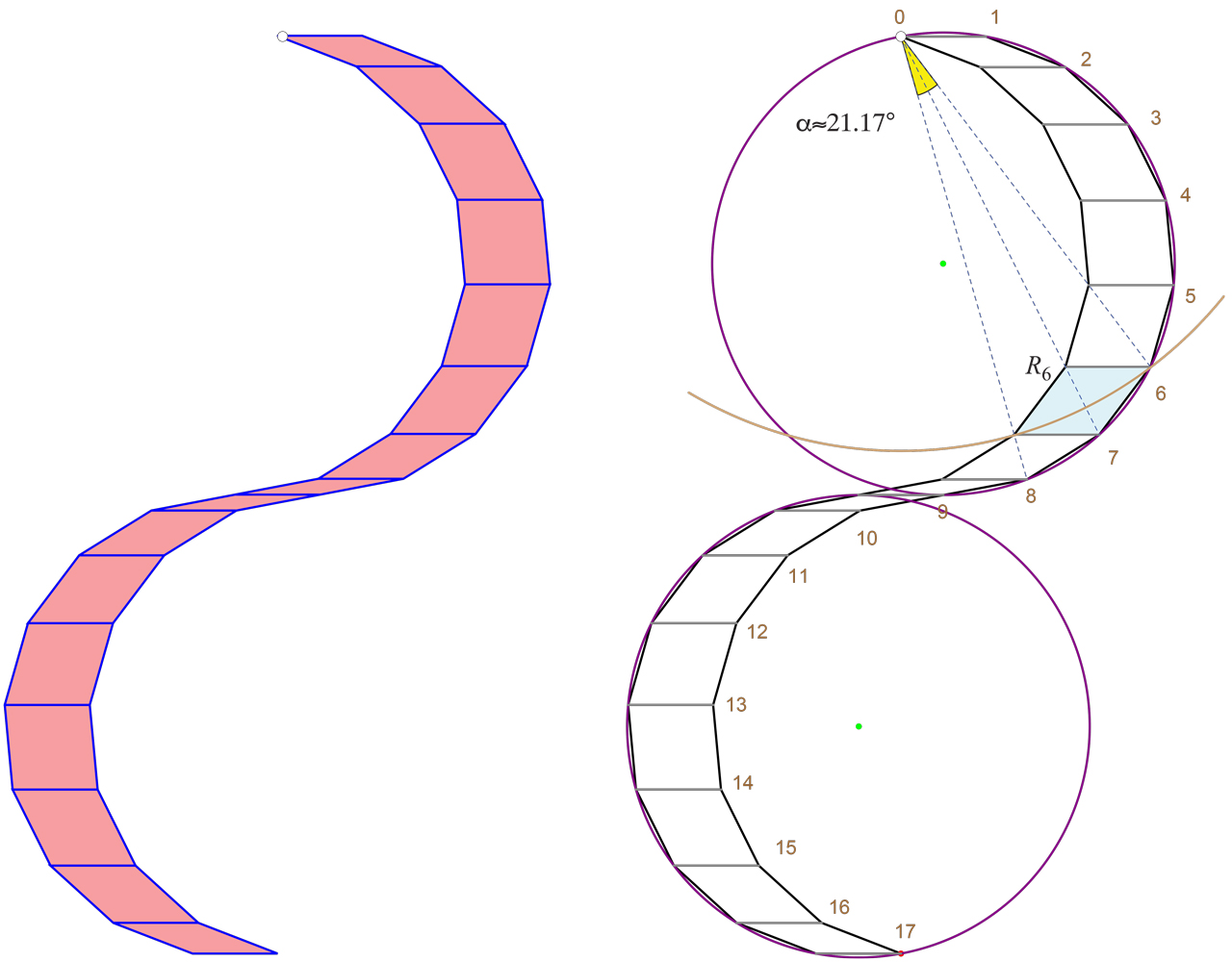}
\caption{$n=17$. (a)~$\q=1^\circ$.
(b)~$\q=0, \a \approx 21.17^\circ$.}
\figlab{S_RegPoly_n17_1deg}
\end{figure}

Let $\partial Z_L$ and $\partial Z_R$ be the right and left boundaries of $Z$,
with $Z$ in the previously described orientation.
Let $Z^+$ and $Z^-$ be the upper and lower halves of the $S$-shape.
Let $P_n$ be the regular polygon of $n$ vertices that passes through the
$n/2$ vertices of $\partial Z^+_R$.
Let $C_R$ be the circle circumscribing $P_n$.\footnote{
With unit rhombs, the radius of  $C_R$ is $1/(2 \sin \a/2)$.}
In Fig.~\figref{S_RegPoly_n16_1deg}(b),
$n=16$ and $\a=2 \pi/n = 22.5^\circ$.

\begin{figure}[htbp]
\centering
\includegraphics[width=1.0\linewidth]{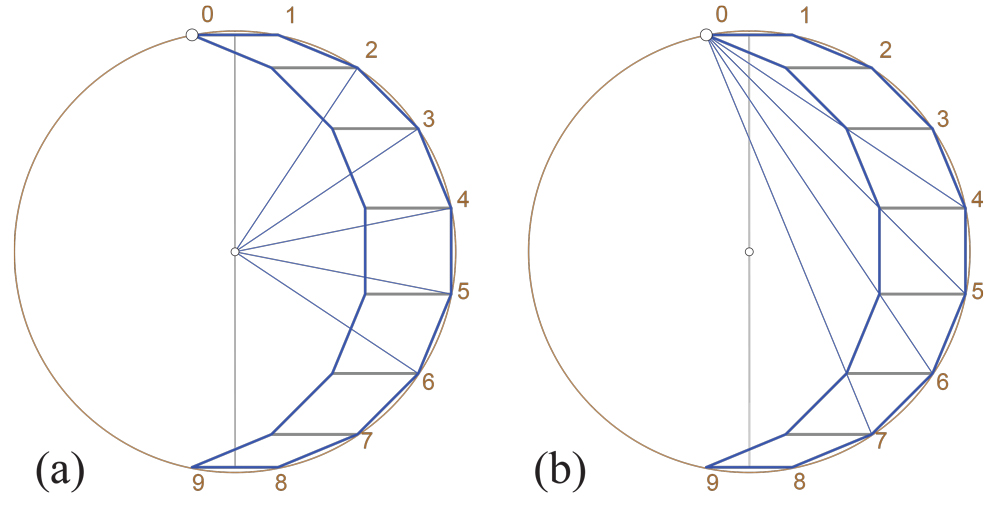}
\caption{$n=16$, $\q=0$, $\a= 22.5^\circ$. 
(a)~Each unit segment inscribed in $C_R$ spans $\a=2\pi/n$ from the center
of $C_R$.
(b)~Each segment spans $\a/2$ form a center on $C_R$.}
\figlab{Zupper_n16}
\end{figure}

We will prove that $| \mathcal{C}(r) \cap Z^+ |$
is exactly $\a$, and $| \mathcal{C}(r) \cap Z^- |$ is less than $\a/2$.
Thus the rotation from $Z$ to $Z'$ leaves $Z^+$ and $Z'^+$ touching,
with a large gap between $Z^-$ and $Z'^-$.
This is evident in Fig.~\figref{Octopus_n12_q05} when $\q=\frac{1}{2}^\circ$.

\begin{lemma}
\lemlab{RegPolyUpper}
$| \mathcal{C}(r) \cap Z^+ | = \a$ for all $r$ intersecting $Z^+$.
\end{lemma}
\begin{proof}
We refer to Figs.~\figref{S_RegPoly_n16_1deg}(b) and~\figref{Zupper_n16} throughout this proof.

Each side of a rhomb along $\partial Z^+_R$ is a unit-length chord of  $C_R$,
and so spans angle $\a$ from the center
of $C_R$; see Fig.~\figref{Zupper_n16}(a).
Therefore each edge spans $\a/2$ from $o=v_0$,
which lies on $C_R$; see Fig.~\figref{Zupper_n16}(b).
In Fig.~\figref{S_RegPoly_n16_1deg}(b), $v_6 v_7$ is such an edge, of rhomb $R_6$.
A full rhomb of $Z^+$ spans a second edge, e.g., $v_8 v_9$.
This determines a chord $c = v_6 v_8$ (red in Fig.~\figref{S_RegPoly_n16_1deg}(b)), which therefore spans $\a$ from $o$.
Thus $\mathcal{C}(r)$ with the appropriate $r$ intersects a rhomb passing through
diagonal endpoints, $R_6$ in our example.

Now we argue for the same conclusion when $\mathcal{C}(r)$ passes through $Z^+$
at an arbitrary location, rather than spanning a diagonal of some specific rhomb $R_i$.
Let $\mathcal{C}(r)$ enter $Z^+$ at $x$ and exit at $y$; see Fig.~\figref{SlidingUpper}.
We now argue that the length of the arc $\mathcal{C}(r) \cap Z^+$ is exactly $\a$, 
just as it is for the diagonal of each $R_i$.
First note that the arc can cross at most two rhombi, i.e., it can cross at most one horizontal
segment of $Z^+$.
Second, note that the arc can be partitioned into two parts, each part
is symmetric with respect to a diagonal.
In the figure, the diagonals are $v_8 o$ and $v_7 o$.
Then the left part of the arc (green) has length $2a$ and the right part (orange)
has length $2b$.
Now notice that $a+b$ is the length of the arc spanned by two adjacent diagonals.
But we know that is $\a/2$. So we have established that $| \mathcal{C}(r) \cap Z^+ | = \a$. 

Thus we proved that $| \mathcal{C}(r) \cap Z^+ | = \a$ for all $r$ crossing $Z^+$.
\end{proof}
\medskip
\begin{figure}[htbp]
\centering
\includegraphics[width=0.75\linewidth]{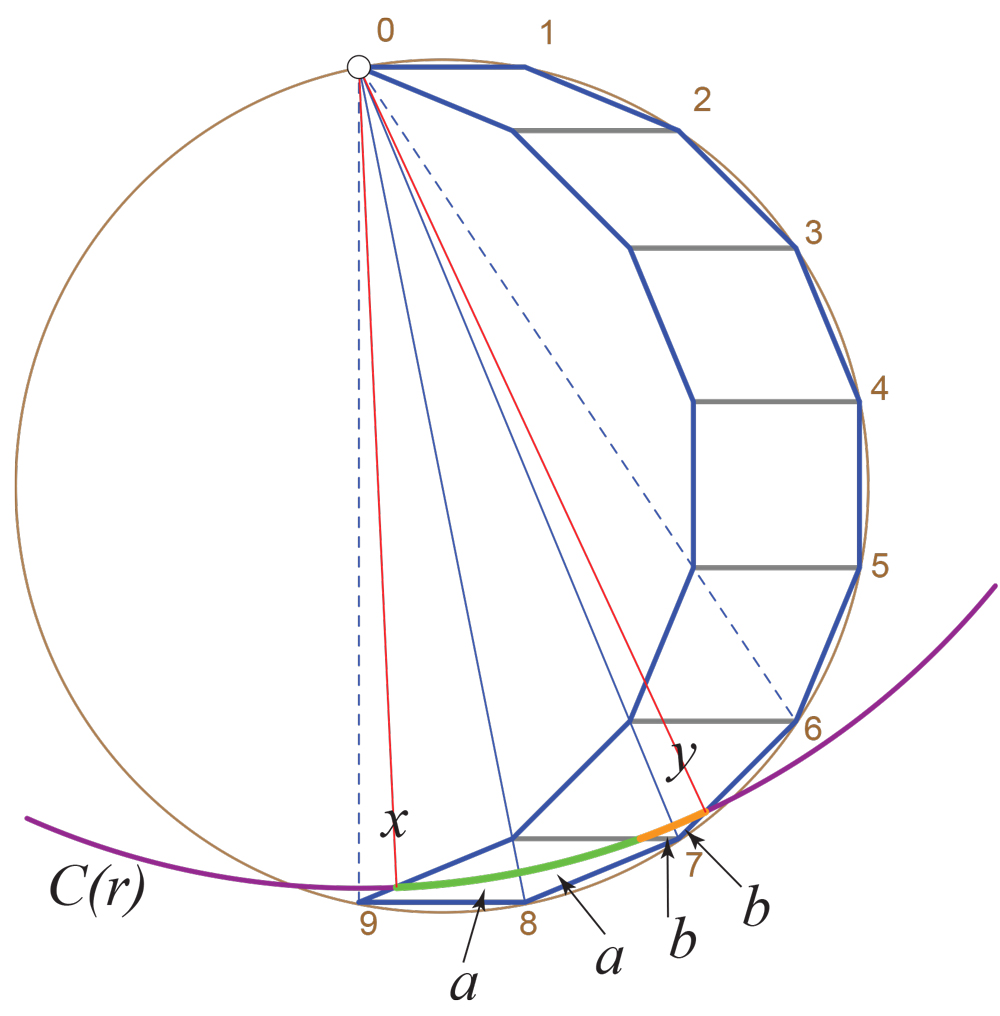}
\caption{$| \mathcal{C}(r) \cap Z^+ | = 2a+2b$.}
\figlab{SlidingUpper}
\end{figure}

Now we turn to the claim for $| \mathcal{C}(r) \cap Z^- |$.
Even though this arc is small, I have not found a simple proof that it is $\le \a$.
In fact, it is strictly less than half of $\a$.
Because my proof is technical, 
and only applies to the degenerate polar zonohedron when $\q=0$,
it is relegated to the Appendix.

\begin{lemma}
\lemlab{RegPolyLower}
$| \mathcal{C}(r) \cap Z^- | < \a/2$ for all $r$ crossing $Z^-$.
\end{lemma}
\begin{proof}
See the Appendix, Section~\secref{Appendix}.
\end{proof}.
\medskip

There remains the situation where $\mathcal{C}(r)$
intersects both $Z^+$ and $Z^-$.
We analyze this case separately in 
Lemma~\lemref{FlatRhomb} in
Section~\secref{FlatRhomb}, which establishes that arc 
$\mathcal{C}(r) \cap Z$ subtends exactly $\a$ from $o$.

\section{Positive $\q$}
\seclab{PositiveTheta}
In comparison to $\q=0$, $Z$ for $\q>0$ is ``stretched out"
in that the turn angle of $\partial Z_L$ between adjacent rhombs
is less than $\a$.
(This turn angle is not constant rhomb-to-rhomb but it is strictly less than $\a$ when $\q>0$.)
This means that $\mathcal{C}(r)$ crosses $Z$ closer to the horizontal,
intersecting rhombs nearer to unit length. Thus the intuition is
that these arcs are short, and become shorter as $r$ increases and pushes the chords
further from $o$.

\begin{figure}[htbp]
\centering
\includegraphics[width=1.0\linewidth]{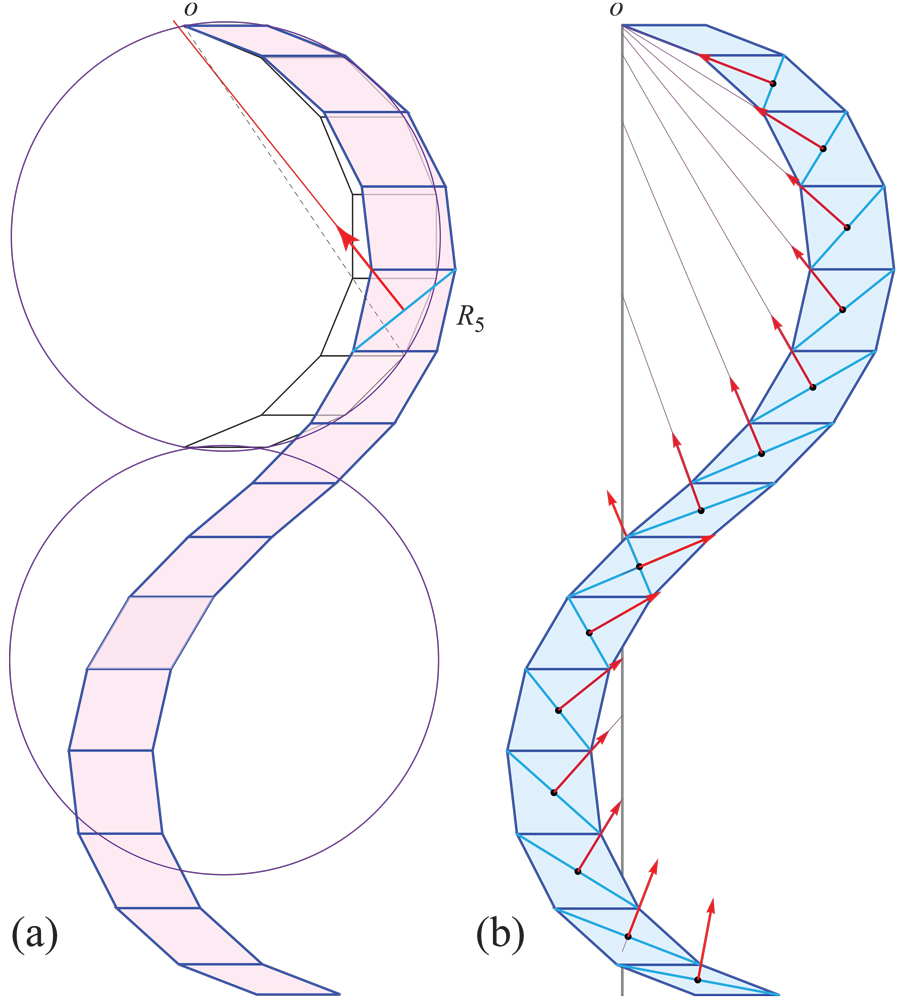}
\caption{$n=16$, $\q=20^\circ$.
(a)~Contrast with $\q=0$.
Diagonal midpoint perpendicular of $R_5$ falls below $o$.
(b)~All diagonal perpendiculars fall below $o$.}
\figlab{Diags_n16_q20}
\end{figure}

\begin{lemma}
\lemlab{PositiveTheta}
$| \mathcal{C}(r) \cap Z | \le \a $ for all $r$ crossing $Z$.
\end{lemma}
\begin{proof}
The proof has four parts:
\begin{enumerate}[label={(\arabic*)}]
\item The chord $c$ determined by $\mathcal{C}(r) \cap Z$ is more shallow (more horizontal) than
rhomb diagonals.
\item If the chord $c$ includes a corner of rhomb $R_i$, then it remains in $R_i$: $c \subset R_i$.
\item If $c$ does not include a corner of a rhomb, then it fits inside the rhomb immediately above or below.
\item The most broad presentation of a rhomb to $o$ is the nearly flat rhomb, with diagonal
near $2$, which subtends strictly less than $2 \cdot \a/2 = \a$ from $o$.
\end{enumerate}
\paragraph{(1).}
Let $c = xy= \mathcal{C}(r) \cap Z$ be the chord of $\mathcal{C}$.
It must be that the perpendicular bisector of $xy$ 
passes through $o$, the center of $\mathcal{C}(r)$.
Recall that the diagonals in the regular-polygon case---$\q=0$,  Fig.~\figref{Zupper_n16}---extend
directly through $o$.
The two diagonals of a rhomb meet orthogonally at their midpoints,
so if $xy$ is one diagonal, the other aims toward $o$.

Fig.~\figref{Diags_n16_q20}(a) contrasts the $\q>0$ situation with $\q=0$.
There, for example, the diagonal of $R_5$, because of the turn-angle reduction/straightening, 
cannot any longer aim toward $o$ (red), as it did when $\q=0$ (dashed).
Extending this to all rhomb diagonals leads to
Fig.~\figref{Diags_n16_q20}(b), which displays all these potential diagonal chords, and
shows that $c$ cannot be a diagonal of $R_i$ for $i \ge 2$:
for each possible rhomb diagonal, the chord perpendicular falls below $o$.\footnote{
Just barely below for $R_2$, by $10^{-6}$.}
\paragraph{(2).}
This shows that the diagonals are too vertically steep 
(in our orientation of $Z$) to serve as the chord $c$:
$c$ must be more nearly horizontal.
This implies that $c$ would stay inside $R_i$ were it to pass through a corner of $R_i$:
$c \subset R_i$.
%
\begin{figure}[htbp]
\centering
\includegraphics[width=0.75\linewidth]{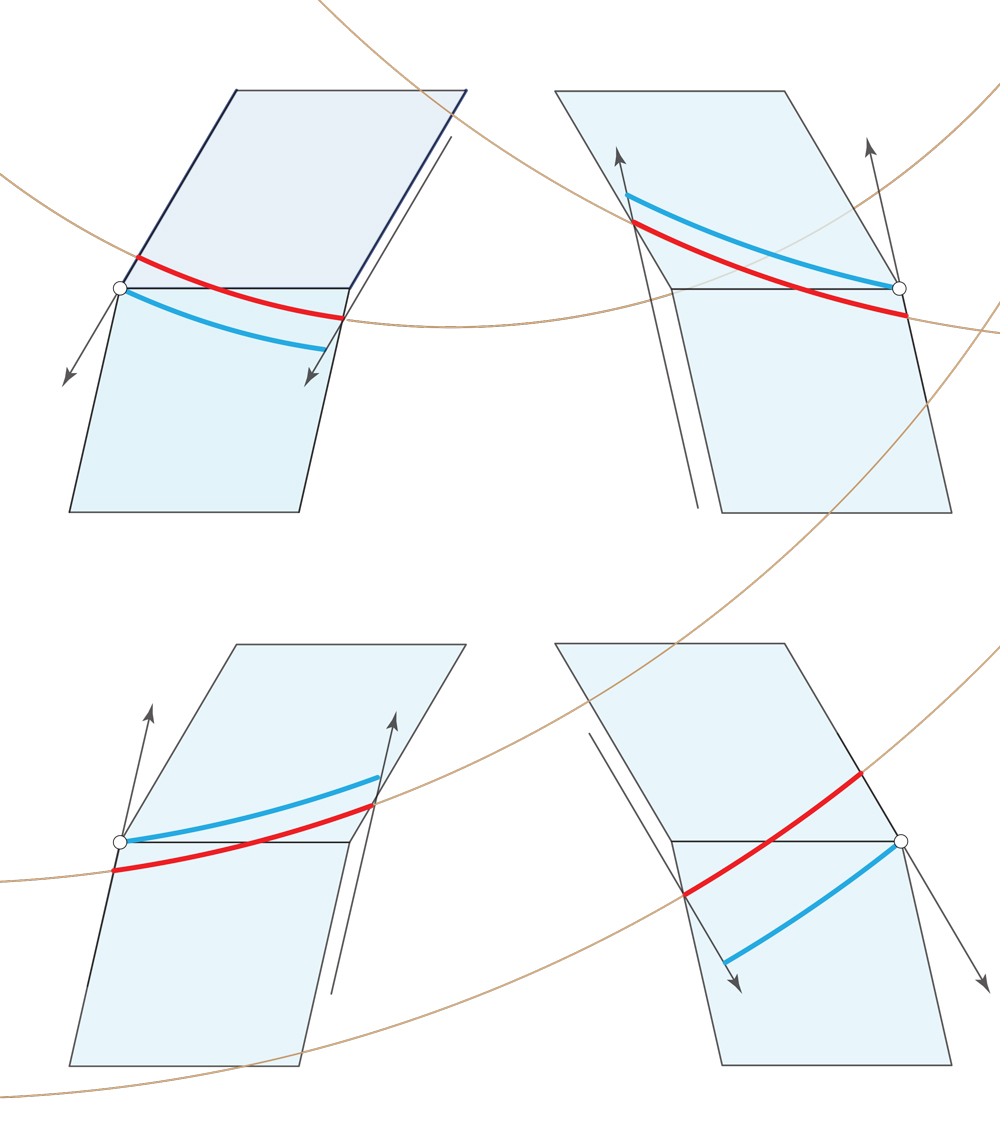}
\caption{The red arc $\mathcal{C}(r) \cap Z$ can be translated up or down to fit
into an adjacent rhomb. The convex corner at an endpoint of the blue translated
arc is marked in each case.}
\figlab{SlideArcs}
\end{figure}

\paragraph{(3).}
First we note that if an arc of $\mathcal{C}(r)$ intersects three or more rhombs,
then with one exception,
its chord $c$ is steeper than the diagonal of a middle rhomb.
The exception concerns the nearly flat rhomb $R_{n/2}$ illustrated
earlier in Fig.~\figref{Disconnected_n16_1deg}.
Because the argument for this exception is not straightforward, we address that in a separate lemma,
Lemma~\lemref{FlatRhomb} in Section~\secref{FlatRhomb} below.
So henceforth assume $c=xy$ crosses a horizontal and intersects two rhombs, $R_i$ and $R_{i+1}$.

Fig.~\figref{SlideArcs} shows that the arc fits in either the adjacent upper or lower rhomb,
by sliding an arc endpoint toward the convex vertex of $\partial Z$.
One endpoint of the horizontal $R_i \cap R_{i+1}$ is a convex vertex and the other endpoint a reflex vertex.
Sliding/translating the arc endpoint $x$ along the side of the rhomb toward the convex vertex
results in the other end of the arc $y$ to lie strictly internal to either $R_i$ or $R_{i+1}$.
This follows because the reflex vertex implies that the extension of the incident $R_i$ rhomb side enters
$R_{i+1}$, and the extension of the incident $R_{i+1}$ rhomb side enters $R_i$.
The only exception is when there is zero turn angle between the rhombs
(which occurs for example with $R_8$ and $R_9$ in Fig.~\figref{S_RegPoly_n17_1deg}),
when the translated arc has both endpoints on rhomb sides.
So in all cases, the arc fits inside a rhomb. In the next step it does not matter in which
rhomb it fits, for we use an upperbound over all rhombs.
\paragraph{(4).}
Now that we know $\mathcal{C}(r) \cap Z$ fits inside some rhomb $R_i$ of $Z$,
we look at the angle subtended at $o$ by the full rhomb $R_i$.
This is an upperbound on the measure of the arc inside
(often a considerably generous upper bound).

Now, the broadest extent of a rhomb from the point of view of $o$ is
the nearly flat rhomb whose long diagonal is nearly $2$ units.
For example, this is the middle rhomb $R_8$ illustrated previously in
Fig.~\figref{S_RegPoly_n16_1deg}(a).
As we know from Fig.~\figref{Zupper_n16}, each
unit side of a rhomb subtends $\a/2$ in $Z^+$ when $\q=0$.
For $\q>0$, each unit side subtends $< \a/2$, not only in $Z^+$ but for all
unit sides in a rhomb of $Z$.
Therefore a whole rhomb subtends $< \a$.

Fig.~\figref{AngsSub} illustrates the actual $\b_i$ subtended angles for $R_i$,
$i=1,\ldots,n-1$, in two
examples. As claimed, $\b < \a$, with equality only for $i=1$.
(The small upticks visible at $R_{12}$ when $n=16$, 
and at $R_{18}$ when $n=24$, reflect the lower tail of the $S$-shape, where
the rhombs become more flattened.)

Equality to $\a$ only occurs when $\mathcal{C}(r)$ crosses the first rhomb $R_1$.
For all larger $r$, the inequality is strict.
\end{proof}
\begin{figure}[htbp]
\centering
\includegraphics[width=0.75\linewidth]{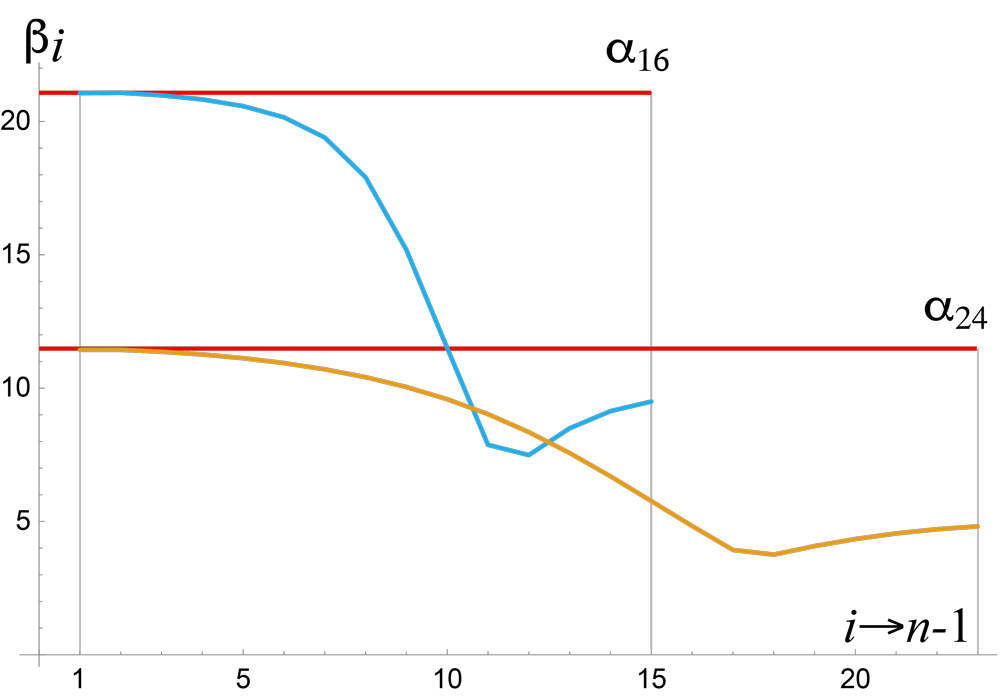}
\caption{$\b_i$ is the angle (in degrees) subtended from $o$ by rhomb $R_i$:
for $n=16$, $\q=20^\circ$ when $\a=22.5^\circ$; and for $n=24$, $\q=40^\circ$,
when $\a \approx 11.48^\circ$.}
\figlab{AngsSub}
\end{figure}

\subsection{Nearly Flat Rhomb $R_{n/2}$} 
\seclab{FlatRhomb}
We encountered a disconnected arc $\mathcal{C}(r) \cap Z$
earlier in Fig.~\figref{Disconnected_n16_1deg}.
This is also an example where the arc crosses three rhombs.
This nearly flat transitional rhomb only occurs when $n$ is even.
Even as $\q \to 0$, no rhomb for $n$ odd becomes arbitrarily close to flat;
see Fig.~\figref{S_RegPoly_n17_1deg}.

In order to complete step~(4) of Lemma~\lemref{PositiveTheta}---to show 
the angle subtended at $o$ by the full rhomb $R_i$ is $<2$---we
analyze this exceptional case in the following lemma.

\begin{figure}[htbp]
\centering
\includegraphics[width=1.1\linewidth]{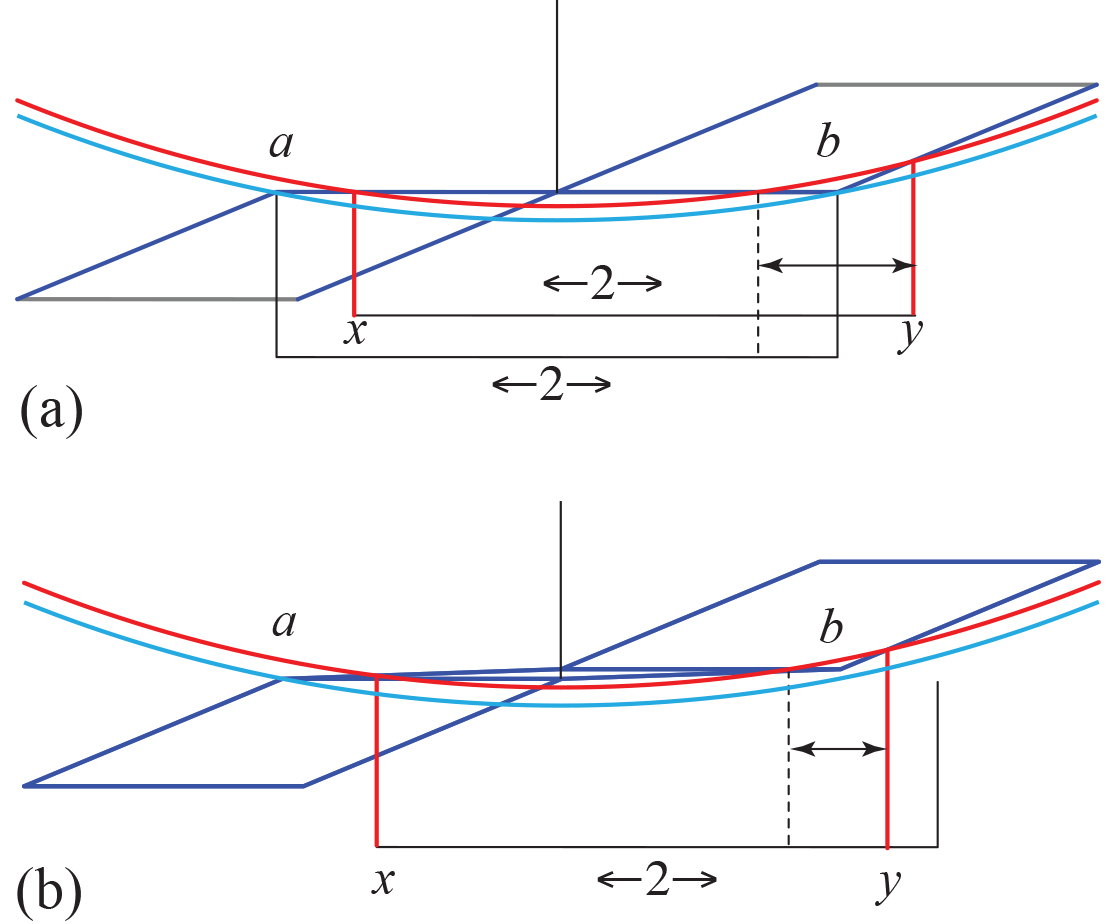}
\caption{$n=16$. (a)~$\q=0$: $|xy|=|ab|$. (b)~$\q=1^\circ$: $|xy|<2$.}
\figlab{Disconnected_n16_lemma}
\end{figure}
\begin{lemma}
\lemlab{FlatRhomb}
If $\mathcal{C}(r)$ intersects $R_{n/2}$, $n$ even, then the chord $c$ of arc $\mathcal{C}(r) \cap Z$ has length $<2$.
\end{lemma}
\begin{proof}
To simplify notation, let $R_0=R_{n/2}$, and let $R_-$ and $R_+$ be the congruent rhombs above and below $R_0$.
Refer to Fig.~\figref{Disconnected_n16_lemma} throughout.

We first analyze the situation when $\q=0$ when $R_0$ reduces to a segment.
If $\mathcal{C}(r)$ passes through the upper corner $a$ of $R_+$, then it also
passes through the lower corner $b$ of $R_-$.
See the blue arc in (a) of the figure.
If $r'<r$ is a bit smaller, then still the 
span of $\mathcal{C}(r') \cap Z = |xy|$ is exactly $2$.
This follows because the portion of the arc in $R_-$ is symmetric about the diagonal line $bo$,
so $|ax|=|by|$.
Note that $\mathcal{C}(r') \cap Z$ (red) is disconnected, and intersects all three rhombs.

Now consider the situation when $\q$ is small, (b) of the figure.
(This is an annotated version of the earlier Fig.~\figref{Disconnected_n16_1deg}.)
$R_0$ is a nearly flat rhomb between $R_-$ and $R_+$.
If $\mathcal{C}(r)$ passes through the upper corner $a$ of $R_+$, it falls below the lower corner $b$ of $R_-$,
because the small intervening height of $R_0$ lifts up $b$.
See the blue arc in (b) of the figure.
If $r'<r$ is a bit smaller, then the right end of the $\mathcal{C}(r')$ arc (red) crosses $R_-$,
but because $b$ is lifted, the span $\mathcal{C}(r') \cap R_-$ is smaller than in the $\q=0$ case.
The result is that the 
span $|xy|$ is strictly less than $2$.

We can quantify the lifting as follows.
Symbolic calculation show that the angle of the central rhomb $R_0$ is exactly $2 \q$;
so $2^\circ$ in Fig.~\figref{Disconnected_n16_1deg}(b).
Let $d = |ax| \le 1$ be the horizontal distance from $a$ to $x$. Then
the entry of $\mathcal{C}(r')$ at $x$ into $R_0$ is lifted by $d \sin \q$, whereas the lifting of $b$ is at least twice
that: $2 \sin \q$.

Although we have ignored the slight deviation from the horizontal of
the bottom and top edges of $R_-$ and $R_+$
when $\q$ is positive,
that deviation can be removed by a slight rotation, leaving the key inequality $|xy|<2$ intact.
\end{proof}

\medskip

\noindent
The slight deviation from the horizontal noted above implies that the nearly vertical diagonal
of $R_0$ does not aim through $o$, and so the diagonal argument used in
Lemma~\lemref{PositiveTheta} still holds for $R_{n/2}$.
This is discernible in the enlargement of $R_{n/2}$ 
shown in Fig.~\figref{S_RegPoly_n16_1deg}(a).

\medskip

With Lemma~\lemref{FlatRhomb}
completing
Lemma~\lemref{PositiveTheta},
we have now proved Theorem~\thmref{net}
via Corollary~\corref{NoOverlap}
for all polar zonohedra for $\q>0$.
And Lemmas~\lemref{RegPolyUpper} and~\lemref{RegPolyLower} settle it for the degenerate polyhedron when $\q=0$.

\section{Open Problem}
\seclab{Open}
That the zone-by-zone unfolding avoids overlap seems almost obvious,
yet the proof presented of Theorem~\thmref{net} feels labored.
Surely there is a simpler proof.

\paragraph{Acknowledgements.}
My unfolding software employs code from Hart's supplement to~\cite{HartPZono}.
I benefitted from discussions with Boris Aronov, Richard Mabry, and Joseph Malkevitch.

\clearpage
\section{Appendix: Lower Regular $n$-gon}
\seclab{Appendix}
In order to prove
Lemma~\lemref{RegPolyLower},
we first abstract the $S$-shaped regular polygon to
circles, surrounding $Z^+$ and $Z^-$, passing through $\partial Z_L$ and $\partial Z_R$.
We call these circles $C^\pm_L$ and $C^\pm_R$.
We compute arc lengths through these circles.
These arc lengths are clearly an upperbound on the arcs intersecting
$Z$ sandwiched between the circles.

Secondly, for computational convenience, we reflect $Z^-$ across the vertical so that all arcs
are to the right of the vertical.
The key result is a ``crescent lemma" that is perhaps known under another name,
as it is pure elementary geometry.

\subsection{The Crescent Lemma}
Refer to Fig.~\figref{TwoCirclesProof} throughout the proof below.
\begin{lemma}
\lemlab{Crescent}
Let $\mathcal{C}(r)$ centered on $o$ intersect $C^-_L$ and $C^-_R$,
the two circles circumscribed and inscribed about $Z^-$. 
Let $S$ be the crescent shape between $C^-_L$ and $C^-_R$
(shaded blue in the figure).
Then $| \mathcal{C}(r) \cap S |$, the angle subtended from $o$ across $S$,
is a constant independent of $r$.
\end{lemma}
\begin{proof}
Incident to $o$ we label four rays $A,B,C,D$.
The arc $ \mathcal{C}(r) \cap S$ subtends an angle $\b$ at $o$, between $B$ and $D$.
To simplify notation, we will use $\angle BD$ to indicate that angle, and similarly for the other
labeled rays.

We first define three of the four rays.
\begin{itemize}
\item $A$ is the vertical ray through $o$ and $C^-_L \cap C^-_R$.
\item $B$ is the ray through the point where $\mathcal{C}(r)$ crosses through $C^-_L$.
\item $D$ is the ray through the point where $\mathcal{C}(r)$ exits through $C^-_R$.
\end{itemize}
Let $\b=\angle BD$ and $\g=\angle AB$.
Define the fourth ray $C$ so that $\angle CD = \g$.

Rotate $BD$ rigidly clockwise about $o$ by $\g$.
By definition, this maps $B$ to $A$, and maps $D$ to $C$.
So now we see that $\b = \angle AC$.

Now consider any other $r'$ with $\mathcal{C}(r')$ intersecting $A$ and $C$.
Then the arc between $A$ and $C$ for this $r'$ is exactly $\b$,
for this arc is just an arc of the circle $\mathcal{C}(r')$,
concentric with $\mathcal{C}(r)$, delimited by $\angle AC$.
Therefore, $| \mathcal{C}(r) \cap S | = \b$ for all $r$ crossing the crescent.
\end{proof}
\begin{figure}[htbp]
\centering
\includegraphics[width=0.75\linewidth]{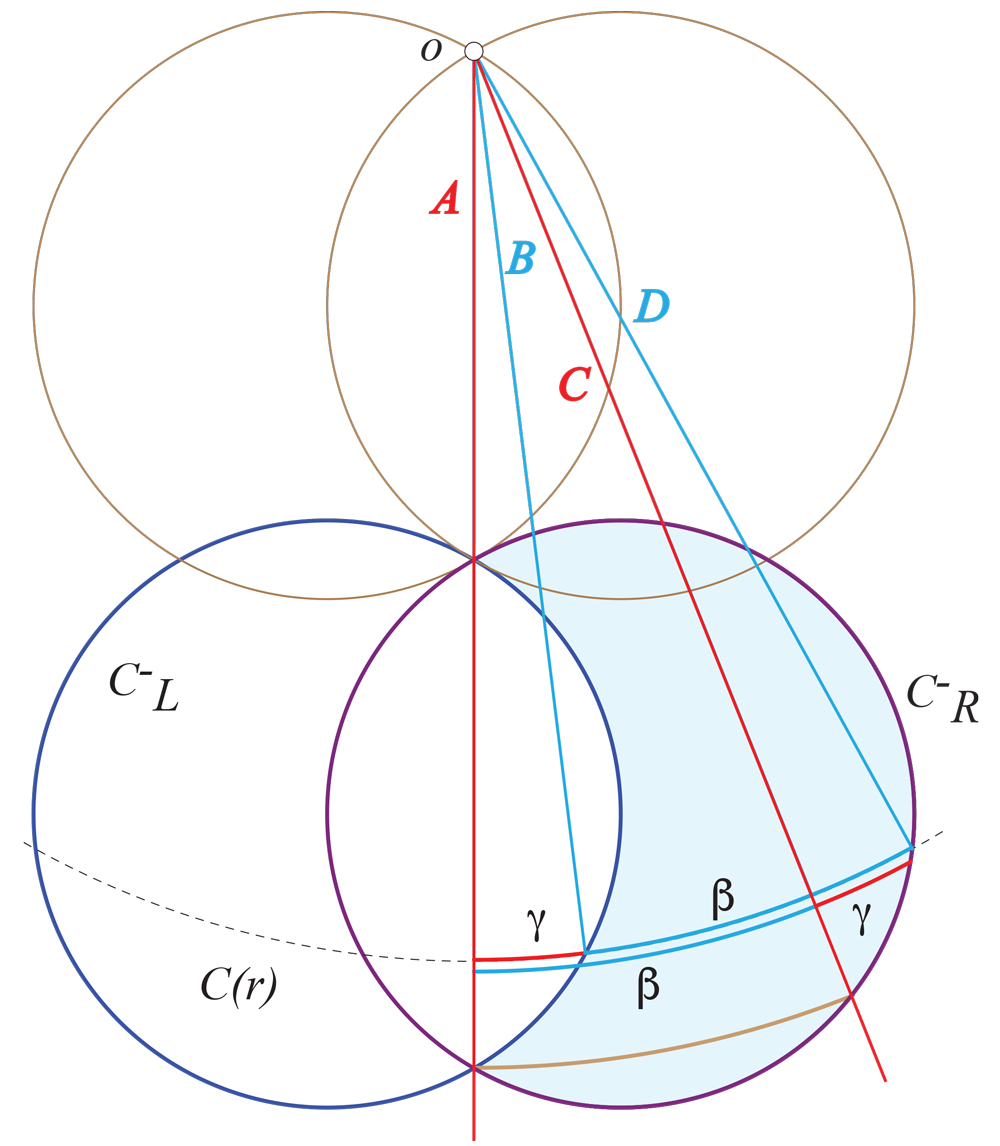}
\caption{$Z^-$ lies in the crescent between $C^-_L$ and $C^-_R$.
$\mathcal{C}(r)$ intersects the crescent in an arc of angle $\b$ independent of $r$.}
\figlab{TwoCirclesProof}
\end{figure}

We label the next result as a Proposition rather than a Lemma, because
we rely on a bound calculated from
an explicit equation.
\begin{prop}
\proplab{CrescentAlpha}
The angle $\b = \angle BD$ defined in
Lemma~\lemref{Crescent}
is strictly less than $\a/2$.
\end{prop}
\begin{proof}
Because all the arc $| \mathcal{C}(r) \cap S |$ are the same length $\b$,
it suffice to measure the length of one, say the vertically lowest arc
shown in Fig.~\figref{TwoCirclesProof}.
We would like to compare this $\b$ to $\a$, which is the counterpart
in $Z^+$.
These angles depend on $n$, which determines the length of each rhomb side.
(Here we use unit radius circles $C^-_L$ and $C^-_R$, so the rhomb sides have
length $L$ as a function of $n$.)
Even though $n$ is a natural number, in the context of circles replacing the discrete rhombs,
we can calculate that
\begin{align*}
\a &= 2 \pi / n \\
L &= 2 \sin( \a/2 ) \\
n &= \pi / \arcsin( L/2 )
\end{align*}
The last equation turns $n$ into a continuous variable.

Using this, I calculated the ratio $\b / \a$ as a function of this continuous $n$.
The result is shown in Fig.~\figref{AlphaRatio}.
\begin{figure}[htbp]
\centering
\includegraphics[width=0.75\linewidth]{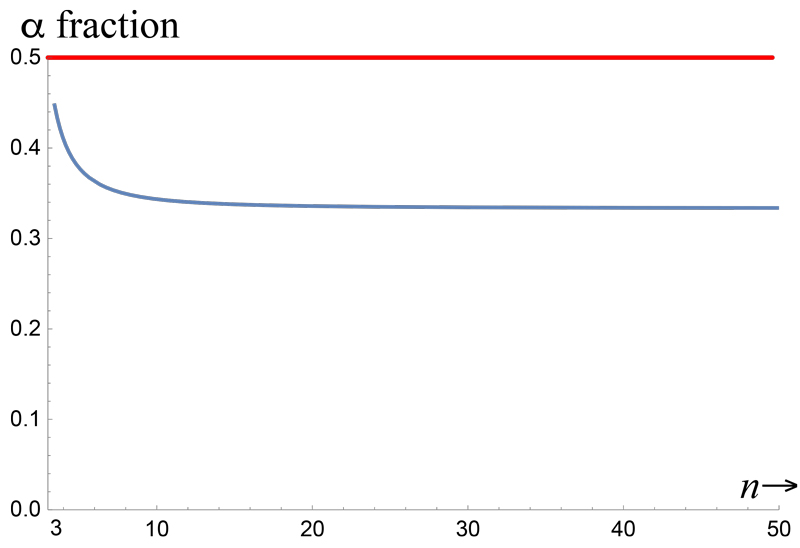}
\caption{The ratio $\b/\a$ for different values of $n$.}
\figlab{AlphaRatio}
\end{figure}
For all $n \ge 3$, $\b / \a < \frac{1}{2}$.
It is evident that as $n \to \infty$, the fraction diminishes to $\frac{1}{3}$,
and as $n$ gets small, the fraction increases but does not reach $\frac{1}{2}$ by $n=3$,
the smallest $n$.\footnote{For $n=3$, the polar zonohedron is a combinatorial cube.}

A formal proof that $\b / \a < \frac{1}{2}$ is left for future work, but in light
of this plot and the equations below, there can be little doubt it holds.
\end{proof}

We have established Lemma~\lemref{RegPolyLower}:
$| \mathcal{C}(r) \cap Z^- | = \b < \a/2$.

\subsection{Equations}
Using $L$ to represent the length of a rhomb side when
the radii of $C^-_L$ and $C^-_R$ are both $1$,
the equations that lead to the plot in Fig.~\figref{AlphaRatio}
are explicitly:

\begin{equation}
\tan \a = \frac{L \sqrt{4-L^2}}{2-L^2} \;.
\end{equation}

\begin{equation}
\tan \b =  \frac{3 \left(L^4-4 L^2-\sqrt{4-L^2} \sqrt{L^4 \left(4-L^2\right)}\right)}
   {L \left(9 \sqrt{4-L^2}
   L^2-36
   \sqrt{4-L^2}+\sqrt{L^4 \left(4-L^2\right)}\right)}
\;.
\end{equation}

\begin{equation}
\b/\a = \frac{\arctan(\tan \b)}{ \arctan (\tan \a)} \;.
\end{equation}

Calculations with these equations verify that the graph shown in Fig.~\figref{AlphaRatio} is accurate.
In particular,
a Taylor-series expansion of the ratio $\tan \b / \tan \a$
results in an expression that 
verifies that, as $n \to \infty$ (and so $L \to 0$),
the ratio approaches $\b/\a$ and its limit is indeed $1/3$:
\begin{equation}
\lim_{L \to 0}\frac{\sqrt{L^4}+L^2}{6 L^2} = \frac{1}{3} \;.
\end{equation}

\newpage
\bibliographystyle{alpha}
\bibliography{refs}
\end{document}